\documentclass[12pt,a4paper]{amsart}

\usepackage[latin1]{inputenc}     
\usepackage{amsmath,amssymb,amsthm}
\usepackage{latexsym}
\usepackage{amsfonts}
\usepackage{bbm,bm,dsfont} 
\usepackage{eucal}
\usepackage{color} 
\usepackage{graphicx}
\usepackage{epstopdf}

\numberwithin{equation}{section} 

\theoremstyle{definition}

\newtheorem{theorem}{Theorem}

\newtheorem{pack}{Packing problem}

\newcommand{\mr}[1]{\mathrm{#1}}
\newcommand{\mc}[1]{\mathcal{#1}}

\newcommand{\ms}[1]{\mathsf{#1}}
\newcommand{\mb}[1]{\mathbb{#1}}

\newcommand{\tr}[1]{\mr{tr}[#1]}

\newcommand{\la}{\langle}
\newcommand{\ra}{\rangle}

\newcommand{\hil}{\mathcal{H}}
\newcommand{\id}{\mathbbm{1}} 
\newcommand{\ket}[1]{|#1\ra}
\newcommand{\ketbra}[2]{|#1\ra\la #2|}
\newcommand{\half}{\frac{1}{2}}

\newcommand{\f}{\varphi}

\title{Effective methods for constructing \\ extreme quantum observables}

\author{Erkka Haapasalo}
\email{erkkath@gmail.com}

\author{Juha-Pekka Pellonp\"a\"a}
\email{juhpello@utu.fi}
\address{Turku Centre for Quantum Physics, Department of Physics and Astronomy, University of Turku, FI-20014 Turku, Finland}

\begin{document}

\maketitle

\begin{abstract}

We study extreme points of the set of finite-outcome positive-operator-valued measures (POVMs) on finite-dimensional Hilbert spaces and particularly the possible ranks of the effects of an extreme POVM. We give results discussing ways of deducing new rank combinations of extreme POVMs from rank combinations of known extreme POVMs and, using these results, show ways to characterize rank combinations of extreme POVMs in low dimensions. We show that, when a rank combination together with a given dimension of the Hilbert space solve a particular packing problem, there is an extreme POVM on the Hilbert space with the given ranks. This geometric method is particularly effective for constructing extreme POVMs with desired rank combinations.

\noindent
\newline
\noindent
\end{abstract}
\maketitle


\section{Introduction}\label{sec:introduction}

The measurement outcome statistics of a quantum measurement are given by quantum observables, which are usually identified with {\it positive-operator-valued measures} (POVMs) \cite{kirja,holevokirja}. In this treatise we concentrate on finite-outcome POVMs on finite-dimensional Hilbert spaces. As an observable $\ms M=(\ms M_j)_{j=1}^N$, $\ms M_j\in\mc L(\hil)$, $\ms M_j\geq0$, $j=1,\ldots,\,N$, $\sum_{j=1}^N\ms M_j=\id_\hil$, ($\mc L(\hil)$ standing here for the algebra of (bounded) linear operators on the Hilbert space $\hil$ the unit element of which is denoted by $\id_\hil$) is measured, the outcome $j$ is detected with the probability $p^{\ms M}_\rho(j)=\tr{\ms M_j\rho}$ supposing that the pre-measurement state of the system was $\rho$.

Quantum measurements can be mixed: We may imagine a classically randomized measurement procedure where a measurement corresponding to a POVM $\ms M^{(1)}$ is triggered with relative frequency $t\in[0,1]$ and another measurement corresponding to a POVM $\ms M^{(2)}$ is triggered for the rest of the time. The effective observable measured is the statistical mixture $\ms M=t\ms M^{(1)}+(1-t)\ms M^{(2)}$, i.e.,\ $\ms M_j=t\ms M^{(1)}_j+(1-t)\ms M^{(2)}_j$ for all outcomes $j$. Indeed, the set of POVMs (with a fixed number of outcomes and operating in the same Hilbert space) is convex.

An extreme point within the set of POVMs (an extreme POVM) is free from this classical randomness and thus contains no redundancy caused by mixing different measurement procedures. According to the Kre\u{\i}n-Milman theorem, the set of POVMs is essentially spanned by the convex mixtures of extreme POVMs. This is further reason why, in many applications, it suffices to concentrate on extreme POVMs.

The question arises on how much resources we need to implement extreme POVMs and one such resource aspect is the possible ranks ${\rm rank}\,\ms M_j$ of an extreme POVM $\ms M$. Naturally, the higher the ranks are, the more resources have to be used in implementing the POVM. We are particularly interested in the possible rank combinations and what conditions these combinations have to satisfy. We find a particularly nice sufficient condition on rank combinations to guarantee that an extreme POVM with that particular combination of ranks exists. The condition is not, however, necessary.

The paper is organized as follows: In Section \ref{sec:basics} we review some known necessary and sufficient conditions for extremality of POVMs from \cite{arveson,dlp,parthasarathy99,Pellonpaa11} and introduce some known conditions the ranks of extreme POVMs have to satisfy. We go on in Section \ref{sec:useful} to give methods for finding new extreme POVMs when we have access to a known extreme POVM and consequently find methods of deducing rank combinations assured to be associated with an extreme POVM using rank combinations that are known to correspond to an extreme POVM. Using these techniques, in Section \ref{sec:lowdim} we are able to characterize the rank combinations corresponding to extreme POVMs in low dimensions. Finally in Section \ref{sec:packing} we introduce certain packing problems and show that solutions to packing problems are always rank combinations of an extreme POVM. There remain POVMs associated with rank combinations which do not solve these packing problems. The packing problem, however, provides a nice visual `algorithm' for finding extreme POVMs with a desired combination of ranks.

\section{Extreme POVMs}\label{sec:basics}

In our investigation we concentrate on finite-outcome POVMs on a finite-dimensional Hilbert space $\hil$, $\dim\hil=d$. We fix the orthonormal basis $\{\ket n\}_{n=1}^d$ for $\hil$ for the duration of this paper. An $N$-outcome POVM $\ms M$ on $\hil$ is identified with an $N$-tuple $(\ms M_j)_{j=1}^N$ of positive operators on $\hil$ such that $\ms M_1+\cdots+\ms M_N=\id_\hil$. We denote by $r_j(\ms M)$ the rank of the $j$:th operator $\ms M_j$.

A POVM $\ms M=(\ms M_j)_{j=1}^N$ has a minimal Na\u{\i}mark dilation $(\mc M,\ms P,J)$ consisting of a Hilbert space $\mc M$ of dimension $r_1(\ms M)+\cdots+r_N(\ms M)$, a projection-valued measure (PVM) $\ms P=(\ms P_j)_{j=1}^N$ on $\mc M$ and an isometry $J:\hil\to\mc M$ such that $\ms M_j=J^*\ms P_jJ$ for all $j=1,\ldots,\,N$ and the vectors $\ms P_jJ\f$, $j=1,\ldots,\,N$, $\f\in\hil$, span $\mc M$. 
One particular choice for a minimal dilation can be constructed by giving the spectral decompositions
\begin{equation}\label{eq:spekM}
\ms M_j=\sum_{k=1}^{r_j(\ms M)}\ketbra{f_{jk}}{f_{jk}},\qquad j=1,\ldots,\,N,
\end{equation}
where the set of vectors $\{f_{jk}\}_{k=1}^{r_j(\ms M)}$ is orthogonal for all $j$. Suppose that $\mc M$ is a Hilbert space of dimension $r_1(\ms M)+\cdots+r_N(\ms M)$ with an orthonormal basis $\{e_{jk}\,|\,k=1,\ldots,\,r_j(\ms M),\ j=1,\ldots,\,N\}$ and set up a linear operator $J=\sum_{j=1}^N\sum_{k=1}^{r_j(\ms M)}\ketbra{e_{jk}}{f_{jk}}$. Also define the PVM $\ms P=(\ms P_j)_{j=1}^N$ on $\mc M$ through $\ms P_j=\sum_{k=1}^{r_j(\ms M)}\ketbra{e_{jk}}{e_{jk}}$, $j=1,\ldots,\,N$. It easily follows that the triple $(\mc M,\ms P,J)$ is a minimal Na\u{\i}mark dilation for $\ms M$.

We say that a POVM $\ms M=(\ms M_j)_{j=1}^N$ is {\it extreme} if the condition 
$$
\ms M_j= \half \ms A_j+\half \ms B_j \quad \forall j
$$
yields $\ms A=\ms B=\ms M$ for POVMs $\ms A=(\ms A_j)_{j=1}^N$ and $\ms B=(\ms B_j)_{j=1}^N$. Extremality characterizations of POVMs (amongst others) ultimately go back to \cite{arveson} but, in our case, the most applicable ones are given in \cite{dlp, parthasarathy99}. 
First, the extremality of a POVM $\ms M$ can be characterized using its minimal Na\u{\i}mark dilation $(\mc M,\ms P,J)$: 
\begin{itemize}
\item[(A)] $\ms M$ is extreme if and only if for an operator $D\in\mc L(\mc M)$ the conditions $\ms P_jD=D\ms P_j$ for all $j=1,\ldots,\,N$ and $J^*DJ=0$ imply $D=0$. 
\end{itemize}
Equivalently, 
\begin{itemize}
\item[(B)] $\ms M$ is extreme if and only if the map $\mc D\ni D\mapsto J^*DJ\in\mc L(\hil)$, where $\mc D$ is the algebra of $\big(r_1(\ms M),\ldots,r_N(\ms M)\big)$-block-diagonal operators with respect to any eigenbasis of $\ms P$, is injective. 
\end{itemize}
Yet another equivalent characterization is that:
\begin{itemize}
\item[(C)]
$\ms M$ is extreme if and only if the set 
$
\{\ketbra{f_{jk}}{f_{jl}}\,|\,k,\,l=1,\ldots,\,r_j(\ms M),\ j=1,\ldots,\,N\}
$
is linearly independent, where the vectors $f_{jk}$ form the spectral decomposition in \eqref{eq:spekM}.
\end{itemize}

All PVMs are extreme, as the characterization (C) immediately show, and PVMs are the only commutative POVMs that are extreme \cite{HePe11}. There are, however, plenty of extreme POVMs with no projections in their range. For instance, a minimal informationally complete POVM consisting of rank-1 operators is such. We note that minimal Na\u{\i}mark dilations can be used to characterize extreme POVMs also in the case of infinite dimensional Hilbert space and uncountably many outcomes \cite{Pellonpaa11}. A physically relevant example of such an extreme POVM is the canonical phase observable \cite{HePe09}. 

As observed in \cite{dlp, parthasarathy99}, the above extremality characterizations imply that an extreme POVM $\ms M$ satisfies:
\begin{itemize}
\item[(i)] $r_1(\ms M)+\cdots+r_N(\ms M)\geq d$,
\item[(ii)] $r_1(\ms M)^2+\cdots+r_N(\ms M)^2\leq d^2$,
\item[(iii)] $r_j(\ms M)+r_k(\ms M)\leq d$ for any $j,\,k=1,\ldots,\,N$, $j\neq k$.
\end{itemize}
The condition (i) above is trivial and holds for any POVM since $\ms M_1+\cdots+\ms M_N=\id_\hil$.
The condition (ii) follows immediately from (C). The condition (iii) follows from the fact that the intersection of the supports of two different operators $\ms{M}_j$ and $\ms{M}_k$ must contain only the zero vector.

The previous conditions (i)--(iii) can be seen as necessary conditions for a vector $(m_1,\ldots,\,m_N)$ to be a rank vector of an extreme POVM.
In what follows, we will investigate the opposite question on sufficient conditions, i.e., we study conditions on vectors $(m_1,\ldots,\,m_N)$ that guarantee the existence of an extreme POVM $\ms M=(\ms M_j)_{j=1}^N$ on $\hil$ with $r_j(\ms M)=m_j$ for all $j=1,\ldots,\,N$. For example, it has been shown in \cite{HaHePe12} that for every $N$ satisfying $d \leq N \leq d^2$, there exists an extremal rank-1 POVM with $N$ outcomes.

\section{Useful results}\label{sec:useful}

In this section, we list and prove some important results we need later on. 
In all the following theorems \ref{theor:rank-1lisaa}, \ref{theor:deleterank}, \ref{theor:refinerank}, \ref{theor:multiplyranks}, and \ref{theor:increaserank}, $\ms M=(\ms M_j)_{j=1}^N$ is an extreme POVM on a $d$-dimensional Hilbert space $\hil$ and the vectors $f_{jk}$ give the spectral decomposition of \eqref{eq:spekM} for the components of $\ms M$. The following theorem is a generalization of a result proven in \cite{HaHePe12}.

\begin{theorem}[Adding of rank-1 operators]\label{theor:rank-1lisaa}
Suppose that $r_1(\ms M)^2+\cdots+r_N(\ms M)^2<d^2$. There is an extreme $(N+1)$-outcome POVM $\ms M'$ with $r_j(\ms M')=r_j(\ms M)$ for $j=1,\ldots,\,N$ and $r_{N+1}(\ms M')=1$. 
In particular, for any $N=d,\ldots,\,d^2$ there is an extreme rank-1 POVM, i.e.,\ an extreme POVM $\ms M$ with $r_1(\ms M)=\cdots=r_N(\ms M)=1$.
\end{theorem}

\begin{proof}
We denote by $\mc R$ the operator system (selfadjoint linear subspace of $\mc L(\hil)$) spanned by $\ketbra{f_{jk}}{f_{jl}}$, $k,\,l=1,\ldots,\,r_j(\ms M)$, $j=1,\ldots,\,N$. 
The assumption on ranks implies that there is a nonzero selfadjoint operator $A\notin\mc R$. 
Since $A$ is selfadjoint, it has a decomposition $A=\sum_{s=1}^r\alpha_s\ketbra{\eta_s}{\eta_s}$ with $\alpha_s\in\mb C$. Hence, at least one, let say $\ketbra{\eta_1}{\eta_1}$ of the rank-1 operators $\ketbra{\eta_s}{\eta_s}$, is not in $\mc R$. 
We then denote $R=\big(\id_\hil+\ketbra{\eta_1}{\eta_1}\big)^{-1/2}$, and set
\begin{eqnarray*}
\ms M'_j&=&R\ms M_jR,\qquad j=1,\ldots,\,N,\\
\ms M'_{N+1}&=&R\ketbra{\eta_s}{\eta_s}R.
\end{eqnarray*}
Clearly, $\ms M'=(\ms M'_j)_{j=1}^{N+1}$ is a POVM.

Let $(\mc M,\ms P,J)$ be a minimal dilation for $\ms M$ where $\mc M$ has the orthonormal basis $\{e_{jk}\,|\,k=1,\ldots,\,r_j(\ms M),\ j=1,\ldots,\,N\}$, $\ms P_j=\sum_{k=1}^{r_j(\ms M)}\ketbra{e_{jk}}{e_{jk}}$ for all $j$, and $J=\sum_{j=1}^N\sum_{k=1}^{r_j(\ms M)}\ketbra{e_{jk}}{f_{jk}}$. Denote $\mc M'=\mc M\oplus\mb C$ and pick any unit vector $e_{N+1}$ from the (1-dimensional) orthogonal complement of $\mc M$ within $\mc M'$. Define $\ms P'_j=\ms P_j$, $j=1,\ldots,\,N$, and $\ms P'_{N+1}=\ketbra{e_{N+1}}{e_{N+1}}$, and $J'=JR+\ketbra{e_{N+1}}{R\eta_1}$. Using the invertibility of $R$, it is simple to check that $(\mc M',\ms P',J')$, where $\ms P'=(\ms P'_j)_{j=1}^{N+1}$, is a minimal dilation for $\ms M'$. For each $j=1,\ldots,\,N$, denote the subspace of $\mc M$ spanned by $\{e_{jk}\}_{k=1}^{r_j(\ms M)}$ by $\mc M_j$ and $\mb Ce_{N+1}=:\mc M_{N+1}$ so that $\mc M=\bigoplus_{j=1}^{N+1}\mc M_j$. The commutant of the PVM $\ms P'$ consists of decomposable, or block-diagonal, operators $D=\bigoplus_{j=1}^{N+1}D_j$, $D_j\in\mc L(\mc M_j)$, $j=1,\ldots,\,N+1$. Suppose that $D=\bigoplus_{j=1}^{N+1}D_j$ is decomposable and define $D_0=\sum_{j=1}^N\ms P_jD\ms P_j$ so that $D=D_0\oplus d\ketbra{e_{N+1}}{e_{N+1}}$ with some $d\in\mb C$. It follows that
$$
(J')^*DJ'=R\big(J^*D_0J+d\ketbra{\eta_1}{\eta_1}\big)R.
$$
Since $R$ is invertible, we find that $(J')^*DJ'=0$ if and only if $J^*D_0J+d\ketbra{\eta_1}{\eta_1}=0$. This is equivalent with $J^*D_0J=0$ and $d=0$ because $J^*D_0J\in\mc R$ and $\ketbra{\eta_1}{\eta_1}\notin\mc R$. Because $\ms M$ is extreme, $J^*D_0J=0$ yields $D_0=0$. Thus $D=0$ implying that $\ms M'$ is extreme.

To prove the last claim, pick any orthonormal basis $\{\f_n\}_{n=1}^d$ of $\hil$ and define the observable $\ms M^d=(\ms M^d_j)_{j=1}^d$, $\ms M^d_j=\ketbra{\f_j}{\f_j}$, $j=1,\ldots,\,d$. 
The POVM $\ms M^d$ is a rank-1 PVM and thus extreme. We may construct rank-1 POVMs $\ms M^N=(\ms M^N_j)_{j=1}^N$, $N=d,\ldots,\,d^2$, where for each $N=d,\ldots,\,d^2-1$ the POVM $\ms M^{N+1}$ is obtained by the method of adding a rank-1 outcome to $\ms M^N$ using the technique introduced above. Since $\ms M^d$ is extreme, all of the $\ms M^N$ are extreme.
\end{proof}

We are interested in the possible rank vectors $\vec m=(m_1,\ldots,\,m_N)_d$ of extreme $N$-outcome POVMs on a $d$-dimensional $\hil$. 
We always order the ranks from highest to lowest and indicate repeating values of rank with a subscript, so that, e.g.,\ $(3,2,2,2)_5=:(3,2_3)_5$. 
Moreover, we may neglect rank-1 outcomes since, according to Theorem \ref{theor:rank-1lisaa}, rank-1 outcomes may be freely added as long as the sum of squares of the ranks does not exceed $d^2$. Thus, e.g.,\ $(3,2,2,2,1,1,1,1)_5$ is replaced with $(3,2_3)_5$. 
Furthermore, we say that a POVM $\ms M$ (with $N$ outcomes) is associated with the rank vector $\vec m$ when ordering the vector $\big(r_1(\ms M),\ldots,r_N(\ms M)\big)$ in a descending order and grouping the recurrent ranks as described above and neglecting the possible rank-1 outcomes we obtain $\vec m$. In this situation, we denote $\vec m=:\vec m(\ms M)$.

\begin{theorem}[Operators can be deleted]\label{theor:deleterank}
Suppose that
$$
\vec m(\ms M)=(m^1_{s^1},\ldots,m^R_{s^R})_d
$$
where $m^1,\ldots,\,m^R$ are natural numbers greater than or equal to 2 with $1\leq r\leq R$. There is an extreme POVM $\ms M'$ with
$$
\vec m(\ms M')=(m^1_{t^1},\ldots,m^R_{t^R})_d,\qquad t^r\leq s^r,\quad r=1,\ldots,\,R
$$
i.e.,\ one can delete operators from an extreme POVM and obtain another extreme POVM, possibly adding rank-1 outcomes to ensure normalization if above $\sum_{r=1}^Rt^rm^r<d$.
\end{theorem}

\begin{proof}
We prove the claim by induction; we show that one rank can be deleted. Let us delete the outcome $\ms M_h$, $1\leq h\leq N$, resulting in a subnormalized $N-1$-outcome POVM. Since $\ms M$ is extreme, also the restricted set $\{\ketbra{f_{jk}}{f_{jl}}\,|\,k,\,l=1,\ldots,\,r_j(\ms M),\ j=1,\ldots,\,h-1,\,h+1,\ldots,\,N\}$ is linearly independent. If the rank of $\id_\hil-\ms M_h$ is $d$, We may define the POVM $\ms M'=(\ms M'_1,\ldots,\ms M'_{h-1},\ms M'_{h+1},\ldots,\,\ms M'_N)$ by setting $\ms M'_j=S\ms M_jS$ for all $j\neq h$ where $S=(\id_\hil-\ms M_h)^{-1/2}$. Otherwise we may add rank-1 outcomes to the subnormalized POVM $(\ms M_1,\ldots,\ms M_{h_1},\ms M_{j+1},\ldots,\ms M_N)$ from the compliment of the linear hull of $\{\ketbra{f_{jk}}{f_{jl}}\,|\,k,\,l=1,\ldots,\\r_j(\ms M),\ j=1,\ldots,\,h-1,\,h+1,\ldots,\,N\}$ until the sum $R$ of the components $\tilde{\ms M}_j$ of the resulting non-normalized POVM $\tilde{\ms M}$ is of full rank. Then we may set $\ms M'_j=R^{-1/2}\tilde{\ms M}_jR^{-1/2}$ for all $j$. The new $\ms M'$ can be shown to be extreme in the same way as in the proof of Theorem \ref{theor:rank-1lisaa}.
\end{proof}

\begin{theorem}[Components can be refined]\label{theor:refinerank}
For each $j=1,\ldots,\,N$ and any natural numbers $R_1,\ldots,\,R_N$ and sequences $(n_{j,r})_{r=1}^{R_j}$ of positive natural numbers such that $\sum_{r=1}^{R_j}n_{jr}=r_j(\ms M)$, there exists an extreme POVM $\ms M'=(\ms M'_{jr}\,|\,r=1,\ldots,\,R_j,\ j=1,\ldots,\,N)$ with $R_1+\cdots+R_N$ outcomes such that $\ms M_j=\sum_{r=1}^{R_j}\ms M'_{jr}$ for each $j$, i.e.,\ the components of an extreme POVM can be refined to obtain a new extreme POVM.
\end{theorem}

\begin{proof}
Set
$$
\ms M'_{jr}=\sum_{s=1}^{n_{j,r}}\ketbra{f_{j,s+n_{j,r-1}}}{f_{j,s+n_{j,r-1}}}.
$$
Because $\ms M$ is extreme, also the restricted set
$$
\{\ketbra{f_{j,k+n_{j,r-1}}}{f_{j,l+n_{j,r-1}}}\,|\,k,\,l=1,\ldots,\,n_j,\ r=1,\ldots,\,R_j,\ j=1,\ldots,\,N\}
$$
is linearly independent implying that $\ms M'=(\ms M'_{jr}\,|\,r=1,\ldots,\,R_j,\ j=1,\ldots,\,N)$ is extreme as well.
\end{proof}

\begin{theorem}[Ranks can be multiplied]\label{theor:multiplyranks}
Let $\vec m(\ms M)=(m^1_{n^1},\ldots,m^R_{n^R})_d$. For any $A=1,\,2,\ldots$ there is an extreme POVM $\ms M'$ on an $Ad$-dimensional Hilbert space associated with the rank vector $(Am^1_{n^1},\ldots,Am^R_{n^R})_{Ad}$. This means that the ranks of an extreme POVM can be multiplied to obtain a new extreme POVM operating on a Hilbert space with similarly multiplied dimension.
\end{theorem}

\begin{proof}
Pick an orthonormal basis $\{\ket a\}_{a=1}^A$ for $\mb C^A$ and define a POVM $\ms M'=(\ms M'_j)_{j=1}^N$,
$$
\ms M'_j=\ms M_j\otimes\id_{\mb C^A},\qquad j=1,\ldots,\,N,
$$
on $\hil\otimes\mb C^A$. Since $\ms M$ is extreme, it easily follows that the set
$$
\{\ketbra{f_{jk}}{f_{jl}}\otimes\ketbra{a}{b}\,|\,k,\,l=1,\ldots,\,r_j(\ms M),\ j=1,\ldots,\,N,\ a,\,b=1,\ldots,\,A\}
$$
is linearly independent. Thus $\ms M'$ is extreme as well.
\end{proof}

\begin{theorem}[Increasing a rank]\label{theor:increaserank}
Pick any $h=1,\ldots,\,N$. There is an extreme POVM $\ms M'=(\ms M'_j)_{j=1}^N$ on a $(d+1)$-dimensional Hilbert space with $r_j(\ms M')=r_j(\ms M)$ for all $j\neq h$ and $r_h(\ms M')=r_h(\ms M)+1$.
\end{theorem}

\begin{proof}
Denote $\hil':=\hil\oplus\mb C$ and pick any unit vector $\psi$ from the orthogonal compliment of $\hil$ within $\hil'$. Define the POVM $\ms M'=(\ms M'_j)_{j=1}^N$,
$$
\ms M'_h=\ms M_h+\ketbra{\psi}{\psi},\quad\ms M'_j=\ms M_j,\qquad j\neq h,
$$
on $\hil'$. Note that we view operators on $\hil$ as operators on $\hil'$ by extending them to operate as the zero-operator on the orthogonal complement of $\hil$. Pick complex numbers $\alpha_{jkl}$, $k,\,l=1,\ldots,\,r_j(\ms M')$, $j=1,\ldots,\,N$. We set
\begin{eqnarray*}
\sum_{j=1}^N\sum_{k,l=1}^{r_j(\ms M)}\alpha_{jkl}\ketbra{f_{jk}}{f_{jl}}&+&\sum_{k=1}^{r_h(\ms M)}\big(\alpha_{h,k,r_h(\ms M)+1}\ketbra{f_{hk}}{\psi}+\alpha_{j,r_h(\ms M)+1,k}\ketbra{\psi}{f_{hk}}\big)\\
&+&\alpha_{h,r_h(\ms M)+1,r_h(\ms M)+1}\ketbra{\psi}{\psi}=0.
\end{eqnarray*}
Multiplying this equation from both sides with $P$ and, on the other hand, with $\id_{\hil'}-P$, where $P$ is the orthogonal projection onto $\hil$, one obtains
\begin{equation}\label{eq:eka}
\sum_{j=1}^N\sum_{k,l=1}^{r_j(\ms M)}\alpha_{jkl}\ketbra{f_{jk}}{f_{jl}}=0
\end{equation}
and
\begin{equation}\label{eq:toka}
\sum_{k=1}^{r_h(\ms M)}\big(\alpha_{h,k,r_h(\ms M)+1}\ketbra{f_{hk}}{\psi}+\alpha_{j,r_h(\ms M)+1,k}\ketbra{\psi}{f_{hk}}\big)+\alpha_{h,r_h(\ms M)+1,r_h(\ms M)+1}\ketbra{\psi}{\psi}=0.
\end{equation}
Since $\ms M$ is extreme, \eqref{eq:eka} yields $\alpha_{jkl}=0$ for all $k,\,l=1,\ldots,\,r_j(\ms M)$ and $j=1,\ldots,\,N$. Operating from right and left on \eqref{eq:toka} with vectors $f_{hk}$, $k=1,\ldots,\,r_h(\ms M)$, and $\psi$ and recalling that these vectors are orthogonal, we obtain $\alpha_{h,k,r_h(\ms M)+1}=\alpha_{h,r_h(\ms M)+1,k}=0$ with $k=1,\ldots,\,r_h(\ms M)+1$. Thus $\ms M'$ is extreme.
\end{proof}

Combining the results of theorems \ref{theor:refinerank} and \ref{theor:increaserank}, one immediately obtains the following:

\begin{theorem}\label{theor:lowtohigh}
Whenever there is an extreme POVM $\ms M$ on a $d$-dimensional Hilbert space with $\vec m(\ms M)=(m^1_{s^1},\ldots,m^R_{s^R})_d$, there is an extreme POVM $\ms M'$ on a $(d+p)$-dimensional Hilbert space with $\vec m(\ms M')=(m^1_{s^1},\ldots,m^R_{s^R})_{d+p}$ for any $p=1,\,2,\ldots$.
\end{theorem}

\section{Rank vectors in low dimensions}\label{sec:lowdim}

Using the results above and excluding the rank combinations prohibited by the conditions (i)-(iii) presented in Section \ref{sec:basics}, we may deduce most of the possible rank vectors of extreme POVMs in low dimensions. In each dimension, we omit the trivial observables associated with $(d)_d$. In dimensions 1-4, we are able to completely characterize the possible rank vectors of extreme POVMs but in dimensions 5-7 some open questions remain.
\begin{itemize}
\item[$d=2$] The only non-trivial extreme POVMs are the rank-1 POVMs associated with $(1,1)_2$, $(1,1,1)_2$, and $(1,1,1,1)_2$ whose existence is guaranteed by Theorem \ref{theor:rank-1lisaa}.
\item[$d=3$] Again, we have the POVMs associated with $(1_s)_3$, $s=1,\ldots,\,9$. The only remaining rank vectors associated with extreme POVMs are $(2,1_s)_3=:(2)_3$. Indeed, their existence follows from the existence of $(1_s)_2$ and Theorem \ref{theor:increaserank}.
\item[$d=4$] The existence of extreme POVMs with vectors $(1_s)_4$ is clear. We obtain $(3)_4$ from $(2)_3$ using Theorem \ref{theor:increaserank}. Since extreme POVMs with $(1_4)_2$ exist, Theorem \ref{theor:multiplyranks} implies that extreme POVMs with $(2_4)_4$ exist. From this one can show by refining ranks (and adding rank-1 outcomes) the existence of extreme POVMs with $(2_3)_4$, $(2_2)_4$, and $(2)_4$.
\item[$d=5$] Again, the cases $(1_s)_5$ are clear and $(4)_5$ follow from $(3)_4$ and Theorem \ref{theor:increaserank}. The cases $(3)_5$, $(3,2)_5$, $(3,2_2)_5$, and $(3,2_3)_5$ follow from $(2_s)_4$, $s=1,\,2,\,3,\,4$, and Theorem \ref{theor:increaserank}. Since the existence of POVMs with $(2_s)_4$, $s=1,\,2,\,3,\,4$, has already been shown, Theorem \ref{theor:lowtohigh} counts for the cases $(2_s)_5$, $s=1,\,2,\,3,\,4$. The remaining open case is $(3,2_4)_5$.
\item[$d=6$] All the cases $(m^1_{s^1},\ldots,m^R_{s^R})_d$ dealt with above with $d=1,\,2,\,3,\,4,\,5$ exist in the $d=6$ case (possibly adding rank-1 outcomes) by Theorem \ref{theor:lowtohigh}. The existence of POVMs with $(3_4)_6$ is guaranteed by Theorem \ref{theor:multiplyranks}, given the existence of extreme POVMs with $(1_4)_2$. From this, the cases $(3_3,2)_6$ $(3_2,2_2)_6$, $(3_2,2)_6$, $(3_3)_6$, $(3_2)_6$, and $(3)_6$ follow by refining or deleting ranks and possibly adding rank-1 outcomes. The cases $(2_9)_6$ follow from $(1_9)_3$ by doubling of ranks, and one obtains all combinations of 2 and 1 from this by refining or deleting. The extreme POVMs $(5)_6$ are obtained from $(4)_5$ by increasing the highest rank. The case $(4,2_5)_6$ follows from $(2,1_5)_3$ by doubling and one obtains all possible rank vectors of extreme POVMs involving 4 from this by refining or deleting. Moreover, we have $(3,2_5)_6$ and $(3,2_4)_6$. The remaining problematic case is $(3_2,2_3)_6$.
\item[$d=7$] Almost all the possible rank vectors corresponding to extreme POVMs in dimension 7 (those suggested by conditions (i)-(iii) of Section \ref{sec:basics}) are guaranteed by methods like the ones used above. The new problematic case now is $(3_2,2_6)_7$.
\end{itemize}

\section{Packing problems}\label{sec:packing}

In this section, we identify a rank vector $\vec m=(m^1_{s^1},\ldots,m^R_{s^R})_d$ with a set consisting of $s^r$ squares with sides of length $m^r$, $r=1,\ldots,\,R$. Moreover, we imagine a $d\times d$-box. We consider the following problem:

\begin{pack}[General packing problem]\label{pack:gen}
For any rank vector $\vec m=(m^1_{s^1},\ldots,m^R_{s^R})$ we consider the problem of packing the $s^r$ $m^r\times m^r$-boxes, $r=1,\ldots,\,R$, into $d\times d$-box: 
the individual boxes should not be broken down to pieces and the rank boxes should not overlap when fitted in the $d\times d$-box. We call this as the {\it general packing problem associated with $\vec m$}.
\end{pack}

\begin{center}
\begin{figure}
\includegraphics[scale=0.3]{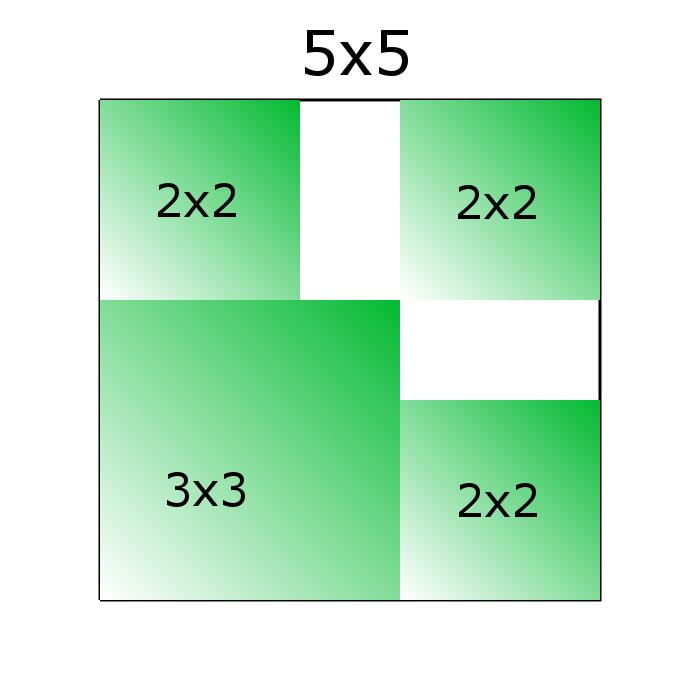}
\caption{\label{fig:5333} As we see, the general packing problem associated with $(3,2_3)_5$ can be solved; the $3\times3$-box and the three $2\times2$-boxes can be fitted inside the $5\times5$-box without overlaps or braking the boxes. After a moments thought, one can see that the packing problem associated with $(3,2_4)_5$ has no solution.}
\end{figure}
\end{center}

Note that if the general packing problem associated with the vector $\vec m$ can be solved, any POVM $\ms M$ with $\vec m(\ms M)=\vec m$ satisfies the conditions (i)-(iii) of Section \ref{sec:basics}. Indeed, the condition (ii) corresponds to the requirement that the area covered by the rank boxes has to fit inside the $d\times d$-box, and if (iii) does not hold, there will be two rank boxes that do not fit inside the big box even if they are set side by side. However, the vector $(3,2_4)_5$ satisfies the conditions (i)-(iii) but, as pointed out in Figure \ref{fig:5333}, the associated general packing problem has no solution. However, numerical calculations show that an extreme POVM associated with this rank vector exists. Thus, solutions to packing problems do not exhaust the set of possible rank vectors of extreme POVMs. However, we will see that a large class of rank vectors of extreme POVMs can be associated to solutions of particular packing problems.

\begin{center}
\begin{figure}
\includegraphics[scale=0.3]{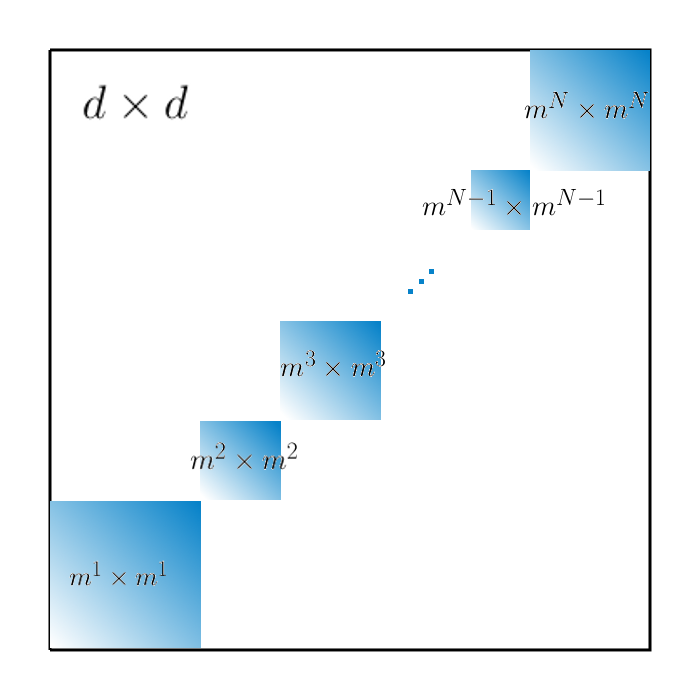}
\caption{\label{fig:pvm} The rank boxes corresponding to the rank vector of a PVM can always be packed inside the $d\times d$-box. In fact, they all fit on the diagonal.}
\end{figure}
\end{center}

In the sequel, we call formation (packing of the rank boxes inside the big box with no overlaps or braking down the boxes) {\it symmetric} when the formation is symmetric with respect to a diagonal through the large $d\times d$-box. We may also formulate the following packing problem:

\begin{pack}[Symmetric packing problem]\label{pack:sym}
If we are able to solve the general packing problem associated with $\vec m$ in a way where the rank boxes can be organized in a formation that can be obtained from a symmetric formation by deleting some of the boxes of the symmetric formation, we say that the {\it symmetric packing problem associated with $\vec m$} has a solution.
\end{pack}

In Figure \ref{fig:5333}, a solution for the general packing problem associated with $(3,2_3)_5$ is given. This is obviously also a solution for the symmetric packing problem. The solution presented in the figure has, in fact a symmetric formation of boxes. Erasing, e.g.,\ the top-left $2\times 2$-box shows that also the packing problem associated with $(3,2_2)_5$ has a solution. The symmetric packing problem associated with $(3_2,2_3)_6$, on the other hand, has no solution, as remarked in Figure \ref{fig:33222}, although the associated general packing problem can be solved. As pointed out in Figure \ref{fig:pvm}, the rank vector of a PVM always has a solution for the general packing problem. Because the rank boxes of a PVM always fit on the diagonal of the $d\times d$-box, the corresponding symmetric packing problem has a solution as well.

\begin{center}
\begin{figure}
\includegraphics[scale=0.25]{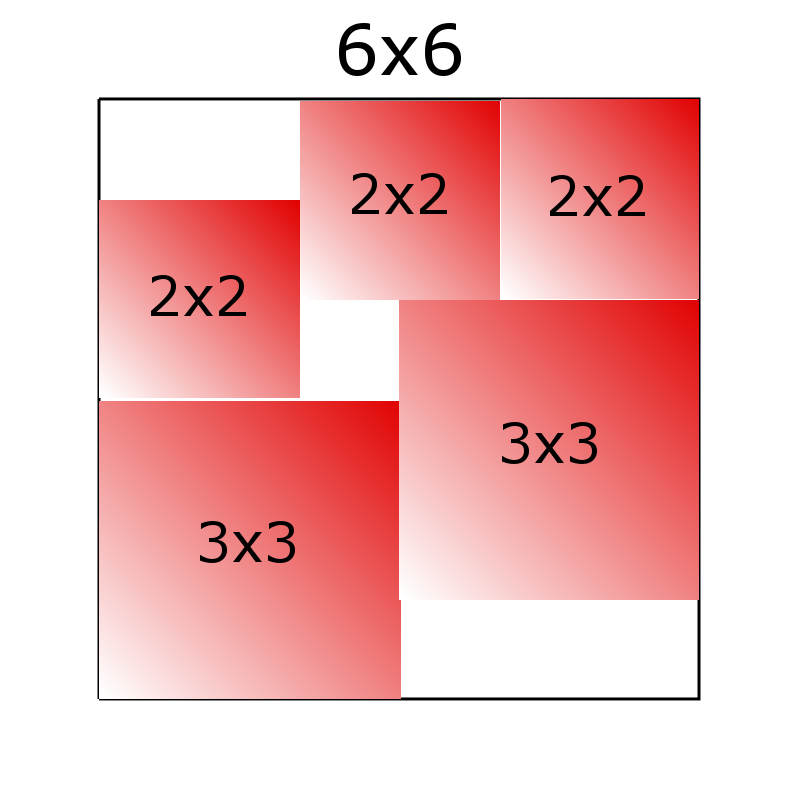}
\caption{\label{fig:33222} Above, one solution for the general packing problem associated with $(3_2,2_3)_6$ is given. However, after a little thinking one finds that the corresponding symmetric packing problem has no solution.}
\end{figure}
\end{center}

In the following theorem, we show that whenever the symmetric packing problem can be solved, we can construct an extreme POVM in the associated dimension with the associated ranks. 
The proof hence gives a recipe for constructing a wide class of extreme POVMs.

\begin{theorem}\label{theor:symext}
Suppose that the symmetric packing problem associated with the rank vector $\vec m$ has a solution. There is an extreme POVM $\ms M$ with $\vec m(\ms M)=\vec m$.
\end{theorem}

\begin{proof}
Fix a finite dimension $d$. Suppose that we are able to solve the symmetric packing problem associated with $\vec m=(m^1_{s^1},\ldots,m^R_{s^R})_d$. Using theorems \ref{theor:deleterank} and \ref{theor:refinerank}, we can show that, if there is an extreme POVM with the rank vector associated with an underlying symmetric solution from which the near-symmetric solution of $\vec m$ can be obtained by possibly deleting rank boxes, there is also an extreme POVM with the rank vector $\vec m$. Hence, we may assume that the rank boxes of $\vec m$ can be arranged into a perfectly symmetric formation.

Fix a $d$-dimensional Hilbert space and an orthonormal basis $\{\ket n\}_{n=0}^{d-1}$ for $\hil$. Let us visualize a $d\times d$ box with $d^2$ slots with the coordinates $(r,s)$, $r,\,s=1,\ldots,\,d$. To each slot $(r,s)$, we associate a vector $g_{rs}\in\hil$,
$$
g_{rs}=\left\{\begin{array}{ll}
\ket r+\ket s,&r>s,\\
\ket r,&r=s,\\
\ket r-i\ket s,&r<s.
\end{array}\right.
$$
According to the note made in the beginning of this proof, we may arrange the rank boxes associated with $\vec m$ into a formation inside this $d\times d$-grid which is symmetric with respect to the diagonal $\{(r,r)\,|\,r=0,\ldots,\,d-1\}$. We assume that the box associated with the rank $m^j$ of $\vec m$ occupies the slots $B^j:=\{(r^j+k-1,s^j+l-1)\,|\,k,\,l=1,\ldots,\,m^j\}$. Because of our assumption, each box $B^j$, $j=1,\ldots,\,N$, $N:=s^1+\cdots+s^R$, is either on the diagonal or completely contained in the upper triangle above the diagonal or in the lower triangle below the diagonal and each box that is not on the diagonal has a pair on the opposite triangle occupying slots with transposed coordinates. From each box $B^j$, we pick the vectors $h_{jk}:=g_{r^j+k-1,s^j+k-1}$, $k=1,\ldots,\,m^j$, associated with the diagonal of the box.

Let us define the operator $R=\Big(\sum_j\sum_{k=1}^{m^j}\ketbra{h_{jk}}{h_{jk}}\Big)^{-1/2}$; if the operator in the parentheses is not of full rank, we may add rank-1 boxes in a symmetric fashion as long as there are empty slots inside the $d\times d$-box. We may now set up a POVM $\ms M=(\ms M_j)_{j=1}^N$,
$$
\ms M_j=R\sum_{k=1}^{m^j}\ketbra{h_{jk}}{h_{jk}}R,\qquad j=1,\ldots,\,N.
$$
Since $R$ is invertible, $\ms M$ is extreme if and only if the set $\{\ketbra{h_{jk}}{h_{jl}}\,|\,k,\,l=1,\ldots,\,m^j,\ j=1,\ldots,\,N\}$ is linearly independent which is what we are going to show next.

\begin{center}
\begin{figure}
\includegraphics[scale=0.3]{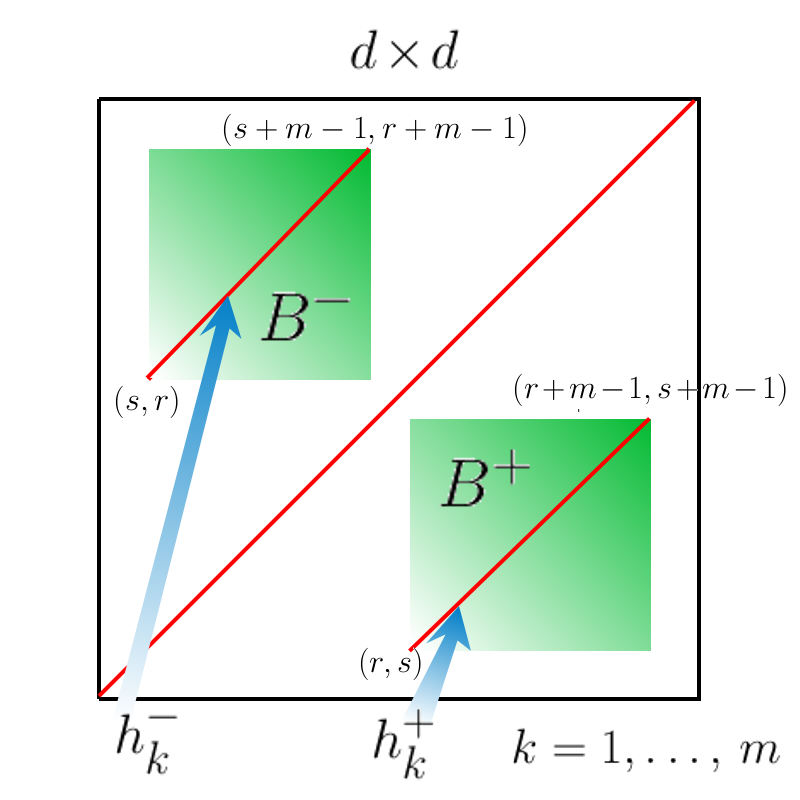}
\caption{\label{fig:todistukseen} }
\end{figure}
\end{center}

It follows easily that, when the box $B^j$ is on the diagonal, the set $\{\ketbra{h_{jk}}{h_{jl}}\,|\,k,\,l=1,\ldots,\,m^j\}$ spans the subalgebra generated by $\ketbra{r^j+k-1}{r^j+l-1}$, $k,\,l=1,\ldots,\,m^j$. Thus, for the boxes $B^j$ on the diagonal, the rank-1 operators $\ketbra{r^j+k-1}{r^j+l-1}$, $k,\,l=1,\ldots,\,m^j$, generate an algebra of block-diagonal matrices in the basis $\{\ket n\}_{n=0}^{d-1}$ whose dimension is the sum of the areas of these boxes. We may hence concentrate on the boxes in the upper and lower triangles.

Let us pick a box $B^+=\{(r+k-1,s+l-1)\,|\,k,\,l=1,\ldots,\,m\}$, $r>s$, in the lower triangle with its twin box $B^-=\{(s+k-1,r+l-1)\,|\,k,\,l=1,\ldots,\,m\}$ in the upper triangle; see Figure \ref{fig:todistukseen}. As always, pick the vectors $h^+_k=g_{r+k-1,s+k-1}$, $h^-_k=g_{s+k-1,r+k-1}$, $k=1,\ldots,\,m$, from the diagonals of these boxes. One has
\begin{eqnarray*}
\ketbra{h^+_k}{h^+_l}&=&\ketbra{r+k-1}{r+l-1}+\ketbra{r+k-1}{s+l-1}+\ketbra{s+k-1}{r+l-1}\\
&+&\ketbra{s+k-1}{s+l-1},\\
\ketbra{h^-_k}{h^-_l}&=&\ketbra{r+k-1}{r+l-1}+i\ketbra{r+k-1}{s+l-1}-i\ketbra{s+k-1}{r+l-1}\\
&+&\ketbra{s+k-1}{s+l-1},\\
k,\,l&=&1,\ldots,\,m.
\end{eqnarray*}
We find that this is the only twin box pair able to produce linear combinations of $\ketbra{h_{jk}}{h_{jl}}$ with contributions from the lower-right block generated by $\ketbra{r+k-1}{s+l-1}$, $k,\,l=1,\ldots,\,m$, and from the upper-right block generated by $\ketbra{s+k-1}{r+l-1}$, $k,\,l=1,\ldots,\,m$. Moreover, we are able to produce linear combinations where only one of these rank-1 operators from the lower-right or upper-left block are present, e.g.,\ $\ketbra{h^+_k}{h^+_l}+i\ketbra{h^-_k}{h^-_l}$ for the lower-right and $\ketbra{h^+_k}{h^+_l}-i\ketbra{h^-_k}{h^-_l}$ for the upper-left block. Thus the linear span of
$$
\{\ketbra{h^z_k}{h^z_l}\,|\,z=\pm,\ k,\,l=1,\ldots,\,m\}
$$
is linearly independent of all the linear spans of corresponding rank-1 operators associated to other box twins and boxes on the diagonal. Moreover the dimension of this linear span is the area $2m^2$ of the two boxes. Treating all the box twins in the same way, we have proven the claim.
\end{proof}

\section{Conclusions}

Methods of creating new extreme POVMs from previously known ones have been established and the consequences of these methods for finding possible rank combinations of extreme POVMs have been discussed. In particular, a `geometric' method of establishing extreme POVMs with rank combinations solving a certain packing problem has been introduced. The method appearing in the proof of Theorem \ref{theor:symext} can be used to define a wide variety of novel extreme POVMs.

Numerical evidence provided by Dr. Navascu\'es however reveals that this geometric method does not cater for all the possible rank combinations of extreme POVMs. In particular, the rank vector $(3,2_4)_5$ solves neither the symmetric nor the general packing problem and yet, according to numerics, an extreme POVM with this combination of ranks exists. This means that the necessary and sufficient rules a rank vector has to satisfy for the existence of an extreme POVM with that particular combination of ranks are more subtle than what found here. However, the new class of extreme POVMs that can be established with the methodology presented in this work greatly widens the set of extreme POVMs. Especially the set of PVMs contributes only to a small fraction of the variety in the set of extreme POVMs found here.

\section*{Acknowledgements}

We thank Drs. Teiko Heinosaari and Miguel Navascu\'es for their feedback and comments on the manuscript. Especially Dr. Navascu\'es is recognized for turning the authors' interest towards the question of possible rank combinations of extreme POVMs and for the numerical evidence he has provided.

\end{document}